\documentclass[%
 reprint,
 amsmath,amssymb,
 aps,
]{revtex4-1}

\usepackage{tikz} 
\usepackage{amsmath,amsfonts,amsthm, mathtools}
\usepackage{float}
\usepackage{dsfont}
\usepackage{csquotes}
\usepackage{amssymb}
\usepackage{verbatim}
\usepackage{bm}
\usepackage{hyperref}
\definecolor{darkred}  {rgb}{0.5,0,0}
\definecolor{darkblue} {rgb}{0,0,0.5}
\definecolor{darkgreen}{rgb}{0,0.5,0}
\hypersetup{
  colorlinks = true,
  urlcolor  = blue,         
  linkcolor = darkblue,     
  citecolor = darkgreen,    
  filecolor = darkred       
}

\theoremstyle{definition}
\newtheorem{definition}{Definition}

\newtheorem{lemma}{Lemma}

\newtheorem{theorem}{Theorem}

\newtheorem{remark}{Remark}
\newtheorem*{remark*}{Remark}

\newcommand{\mbb}{\mathbb}
\newcommand{\mc}{\mathcal}

\newcommand{\tr}{\textrm{Tr}}

\usepackage{multirow}
\newcommand{\ket}[1]{|#1\rangle}
\newcommand{\bra}[1]{\langle #1|}

\definecolor{cool_green}{rgb}{0.0, 0.5, 0.0}

\usepackage{blkarray}
\usepackage{graphicx}
\usepackage{dcolumn}
\usepackage{bm}

\begin{document}

\preprint{APS/123-QED}

\title{Summing to Uncertainty: On the Necessity of Additivity in Deriving the Born Rule}

\author{Jiaxuan Zhang}
\email{jiaxuan.zhang@physics.ox.ac.uk}

\affiliation{
Department of Physics, University of Oxford, Oxford, United Kingdom
}

\date{\today}

\begin{abstract}

The emergence of intrinsic probability has long been one of the most important and puzzling problems in quantum mechanics, and the law most directly related to this problem is the Born rule. For a century, there have been many attempts to derive the Born rule as a theorem rather than postulating it. However, existing derivations of the Born rule are each based on different frameworks and have attracted different criticisms. The assumptions from which they start are also highly divergent, and the connections between them have not been sufficiently studied. These possible connections are very likely to be the key to answering questions about the origin of probability in quantum mechanics.

This paper focuses on proving the necessity and indispensability of the additivity assumption in the derivation of the Born rule. This supports the view that the Born rule cannot be derived solely from other non-probabilistic quantum or additional postulates. We first prove that additivity cannot be derived from two other commonly used non-probabilistic additional assumptions, non-contextuality and normalization. Then we analyze the crucial role of the additivity assumption in five important existing derivations of the Born rule. These include Gleason’s Theorem, Busch’s extension of Gleason’s Theorem, the Deutsch–Wallace Theorem, Zurek’s envariance proof, and the Finkelstein–Hartle Theorem. We show that these derivations either depend heavily on the additivity assumption or lead to obvious loopholes due to the lack of additivity. We also point out some problems arised from the lack of a non-contextuality assumption.

Our results provide a novel insight into the important role of additivity assumption in quantum measurement, as well as into the origin of probability in quantum mechanics.

\end{abstract}

\maketitle



\section{Introduction}

Ever since the early development of quantum mechanics, its intrinsic probabilistic nature has been one of the most puzzling problems for physicists. This problem is described by the famous quantum measurement problem. Briefly speaking, the measurement problem concerns the question of why a quantum system evolves continuously and deterministically according to the Schrödinger equation when it is not measured, while a measurement introduces a completely different, discontinuous, and random process \cite{Norsen2017}. The probabilistic result of a measurement is described by the Born rule \cite{Born1926CollisionEN,Nielsen_Chuang_2010}:

$$p(A) = \tr{[\rho A]}$$

where $A$ is a projective operator or a positive operator-valued measure (POVM) effect. Also, if $A$ is an observable, the Born rule gives us the expectation value for the measurement \cite{Nielsen_Chuang_2010}.

In the standard formalism of quantum mechanics, the Born rule is introduced as the last axiom, or the last postulate \cite{dirac1958principles,vn2018,Nielsen_Chuang_2010,Zurek_2018}. Unlike the two simple axioms of special relativity, which have clear physical intuitions, the quantum postulates are often criticized for being complicated and lacking transparent physical interpretations \cite{Caves_2005_Properties_of_the_frequency_operator,Paw_owski_2009}. For example, why should quantum states be represented by vectors in Hilbert spaces? Why are quantum measurements represented by projective operators onto subspaces of Hilbert spaces? Most importantly, why must we postulate the Born rule, a second dynamical law that appears to be unrelated to the Schrödinger equation? And where does its probabilistic nature originate? The goal of all existing interpretations of quantum mechanics is to address this measurement problem \cite{Norsen2017}. The Born rule can therefore be regarded as a core of the measurement problem. As a result, there have been many attempts to derive it rather than postulate it \cite{gleason1957_original_GT,busch2003,deutsch1999,zurek2005,Zurek_2018,hartle2019quantummechanicsindividualsystems,Finkelstein1963}.Some of these derivations also provide, to a certain extent, reasonable explanations for other quantum postulates, like the ones for quantum states \cite{gleason1957_original_GT} and measurement observables \cite{Zurek_2018}.

It has long been desired by many researchers that the Born rule, which is the only explicitly probabilistic postulate in quantum mechanics, could be derived from the other non-probabilistic quantum postulates. This expectation is particularly strong among supporters of deterministic quantum interpretations, such as the Many-Worlds Interpretation (MWI) \cite{everett1957,Finkelstein1963,hartle2019quantummechanicsindividualsystems,Zurek_2018}. Unfortunately, this attempt does not appear to be successful. All existing derivations of the Born rule involve the introduction of additional assumptions beyond standard quantum mechanics, even without considering the debates over the validity of each derivation. Vaidman has discussed several important derivations in \cite{vaidman2020_Derivations_Born_Rule} and concludes that extra postulates are unavoidable in deriving the Born rule. While we agree with this argument, we believe that a more detailed analysis of existing derivations and their underlying postulates should be presented to support it. 

In many original works on deriving the Born rule, the assumptions employed are not presented clearly or rigorously enough. As a result, it is difficult to determine which postulates are genuinely additional, which are indispensable for the derivation, and whether some of the extra assumptions can be derived from the non-probabilistic quantum axioms or from other additional assumptions. There are a few authors attempting to combine some additional assumptions --- Logiurato and Smerzi have claimed that the additivity assumption can be replaced by the non-contextuality assumption \cite{Logiurato_2012}. Answering these questions is essential for locating the origin of probability in the Born rule and in quantum mechanics.

The focus of this work is on the necessity and importance of the additivity assumption in the derivation of the Born rule. We will show that the proposed substitution in \cite{Logiurato_2012} is unsuccessful. Additivity is a fundamental property of measures or probability measures in mathematics \cite{folland2013real}, and thus carries probabilistic feature. In contrast, the non-contextuality of measurements does not contain an explicit probabilistic property. Some authors have also argued that the probabilistic nature of the Born rule must ultimately originate from additional assumptions that themselves possess probabilistic features \cite{Barnum_2000_Quantum_probability_from_decision_theory}. In fact, we will prove in this work that among the five different derivations we analyze, there is no equivalent substitute for the additivity assumption. All five derivations require additivity itself. In some cases, additivity is assumed explicitly as a core postulate and plays a central role throughout the proof. In others, we will demonstrate the problems that arise in the absence of additivity and show how these problems can be resolved by introducing an additivity assumption.

The structure of this paper is as follows. In Section \ref{section_preliminaries}, we introduce the background material necessary for our discussion, including the non-probabilistic postulates of quantum mechanics and three commonly used additional assumptions: additivity, non-contextuality, and normalization. In Section \ref{section_additivity_relations}, we analyze the relationships between the additivity assumption and the other two, and prove that no equivalence relation exists among them. In Section \ref{section_additivity_in_5_proofs}, we examine the role of the additivity assumption in five prominent derivations of the Born rule. These include two Gleason-type derivations \cite{gleason1957_original_GT,busch2003} and three MWI derivations \cite{deutsch1999,zurek2005,hartle2019quantummechanicsindividualsystems}. We call them MWI derivations because they were originally proposed by proponents of the Many-Worlds Interpretation and exhibit certain MWI characteristics. However, these three derivations do not rely on MWI and are also valid within the framework of standard quantum mechanics. We will show that many of the problems encountered in these derivations arise from the absence of an explicit additivity assumption. Although some of these issues have already been noted in the literature, our analysis provides a novel and illuminating perspective on their underlying cause.

\section{Preliminaries}
\label{section_preliminaries}

In this section, we introduce the necessary background knowledge for the analyses in this paper. Some comes from quantum mechanics, and the rest comes from measure theory in mathematics. For our purpose, we will adjust the expressions of some definitions from measure theory. We will also give rigorous definitions for three important properties or assumptions: additivity, non-contextuality, and normalization, distinguishing between the different meanings of these concepts that appear in the literature under the same name. This will help us clarify how the apparent equivalence of two concepts can emerge from carelessness in definitions.

\subsection{Non-Probabilistic Quantum Postulates}

The postulates of the standard quantum formalism first come from the works of Dirac \cite{dirac1958principles} and von Neumann \cite{vn2018}. Combined with the definitions used in modern textbooks and literature \cite{Nielsen_Chuang_2010,Zurek_2018}, we list all the non-probabilistic quantum postulates as follows:

(QM1). \textbf{State postulate} --- A pure quantum state is represented by a unit vector $\ket{\psi}$ in a Hilbert space. A mixed state is represented by a convex combination of the density matrices of the pure states.

(QM2). \textbf{Evolution postulate} --- Quantum states evolve under unitary transformations, or equivalently, according to the Schrödinger equation.

(QM3). \textbf{Observable postulate} --- The measurement operator of any measurable physical property can be described by an observable

$$\Omega_A = \sum_i a_i A_i$$

where $A_i = \ket{a_i}\bra{a_i}$ is the projective operator of the eigenstate $\ket{a_i}$, and $a_i$ is the corresponding eigenvalue.

We can also denote the measurement operators as a set $\{A_i,a_i\}_i$. This indicates that we can perform the measurement for a single projective operator without specifying the entire set. This raises the question of contextuality, which we will discuss later.

More general quantum measurements are not only described by projections, but also by arbitrary positive semidefinite operators. Therefore, the Observable postulate has a variant:

(QM3'). \textbf{POVM postulate} --- A general quantum measurement is represented by a positive operator-valued measure (POVM):

\begin{definition}

A POVM $\{A_i\}$ is a set of operators, or POVM effects $A_i$, that satisfies the following properties:

(i) For all $i$, $A_i \ge 0$.

(ii) $\sum_i A_i = I$, where $I$ is the identity operator.

\end{definition}

Finally, we have the Eigenstate postulate, which describes a special case for quantum measurement that is deterministic and continuous:

(QM4). \textbf{Eigenstate postulate} --- If a state is in an eigenstate $\ket{a_i}$ of the observable $\Omega_A$, the probability of obtaining $a_i$ when measuring $\Omega_A$ is 1.

\subsection{Measure and Relevant Additional Assumptions}

The Observable postulate (QM3) describes the mathematical expression for the operator of a measurement operation. To obtain the measurement result, we need another function that takes measurement operators and states as input and outputs the outcomes we can observe in experiments. In standard quantum mechanics, this function is given by the Born rule. However, since our purpose is to derive the Born rule, we only assume it to be a real-valued function here. This assumption is reasonable since we never obtain imaginary values from experiments.

(M1). \textbf{Real Function postulate} --- Once we have the state and the measurement operator, the measurement process can be assumed to be a map or a function that maps the operator and the state to the extended real line:

$$\mu(\rho,A) = r \in \mbb{R} \cup \{\pm \infty\}$$

where $\rho$ is a quantum state and $A$ is a positive semidefinite operator.

Usually, $\mu$ is assumed to be a measure in the literature on derivations of the Born rule \cite{gleason1957_original_GT, busch2003}, which directly introduces many important properties, such as additivity and non-negativity. However, since our goal is to study the necessity of each assumption, we only assume $\mu$ to be a real-valued function. We include $\{\pm \infty\}$ in the codomain, since infinite values are allowed by the definition of a measure \cite{folland2013real}. This can cause problems related to continuity, which we will discuss in Section \ref{section_additivity_in_5_proofs}. Also, this output real number is not assumed to be a probability.

We now discuss some special properties that this function may have. The definitions of additivity, non-contextuality, and normalization are given below, each with two different versions. This division may seem tedious at first, but it will be shown to be essential for clarifying the relationships between these properties or assumptions.

\begin{definition}

\cite{gleason1957_original_GT,busch2003} The function $\mu$ is countable additive if for any countable set $\{A_i\}$ of mutually compatible operators $A_i$'s, 

$$
\mu(\sum_{i=1}^{\infty} A_i) = \sum_{i=1}^{\infty} \mu (A_i)
$$

Here, mutual compatibility means that the $A_i$ are mutually disjoint if the $A_i$ are projections, and that $\sum_i A_i \leq I$ if the $A_i$ are POVM effects.

\end{definition}

This is the additivity assumed by Gleason and Busch in their derivations of the Born rule \cite{gleason1957_original_GT,busch2003}. We will show below that this additivity usually (but not necessarily) involves an implicit non-contextuality assumption.

\begin{definition}

\cite{folland2013real} The function $\mu$ is finite additive if for any finite set $\{A_i\}$, it satisfies

$$
\mu(\sum_{i=1} A_i) = \sum_{i=1} \mu (A_i)
$$

\end{definition}

\begin{definition}

The measurement function $\mu$ is observable non-contextual (ONC) if the measurement result $\mu(A)$ of a measurement operator $A$ is independent of the operation of other compatible measurements.

This is the usual meaning of non-contextuality in the literature \cite{bell1966,Logiurato_2012}. The compatible measurements are referred to as the context of the measurement. Here, $A$ and the compatible measurements can be projective operators, POVM effects, or observables.

\end{definition}

\begin{definition}

The measurement function $\mu$ is additivity non-contextuality (ANC) if it has an additivity property, and this additivity property is independent of the context of the measurement.

\end{definition}

For example, let $\{A,B,C_1,C_2\}$ and $\{A,B,D_1,D_2\}$ be two different sets of measurements. The operators in both sets sum to $I$, and $\{C_1,C_2\} \ne \{D_1,D_2\}$. If ANC does not hold, we may have
$\mu(A+B|\{C_1,C_2\})=\mu(A|\{C_1,C_2\})+\mu(B|\{C_1,C_2\}) \ne \mu(A+B|\{D_1,D_2\}) = \mu(A|\{D_1,D_2\})+\mu(B|\{D_1,D_2\})$.

This property can be regarded as a sub-property of additivity.

\begin{definition}

The function $\mu$ is normalized if $\mu(I)=1$.

\end{definition}

This is the usual definition of normalization for a probability measure \cite{busch2003}. In this work, whenever we refer to normalization, we mean this definition. However, some literature uses the term normalization to refer to another meaning given below.

\begin{definition}

The function $\mu$ is strong normalized if, for any countable set $\{A_i\}$ of mutually compatible operators (projections or POVMs) satisfying $\sum_i A_i = I$, we always have $\sum_i \mu(A_i) = 1$.

\end{definition}

This property is not merely normalization, but a combination of normalization and additivity. It is misleading to call it normalization, but we define it in this way to help clarify the confusion present in some literature \cite{Logiurato_2012}.

Additionally, the State postulate (QM1) also contains a normalization property and a non-negativity property. However, these properties do not imply normalization or non-negativity for the measurement function $\mu$.

\section{Relation between Additivity and Other Assumptions}
\label{section_additivity_relations}

In this section, we discuss the relation between additivity, non-contextuality, and normalization. First, we prove the non-equivalence of additivity and non-contextuality. We present counterexamples of functions $\mu$ that satisfy one property but not the other. We also rephrase the proof of their equivalence presented in \cite{Logiurato_2012}, and explain the problem in it. At the end of this section, we prove that additivity can be derived from strong normalization, but not normalization.

\subsection{Additivity and Non-Contextuality}
\label{subsection_Additivity_and_Non-Contextuality}

Non-contextuality and additivity are two extra assumptions commonly used in different existing derivations of the Born rule. Logiurato and Smerzi claim that the two are equivalent, namely, they can be derived from each other \cite{Logiurato_2012}. We will prove that they are not equivalent. The apparent equivalence arises from ambiguity in the definitions of the assumptions.

First, it is obvious that additivity cannot be derived from ONC. A simple counterexample is

$$\mu(A)=\left(\frac{\tr[A]}{d}\right)^2 $$

where $d$ is the dimension of the Hilbert space to which operator $A$ maps. It is clear that $\mu$ satisfies ONC. However, for rank-one projective operators $A,B$ in a 2-dimensional space,

$$\mu(A+B) = 1 \ne \mu(A)+\mu(B)=\frac{1}{2}$$

Also note that this counterexample satisfies normalization but not strong normalization. The problematic conclusion of \cite{Logiurato_2012} actually arises from using strong normalization as the normalization assumption, which already implies additivity (see Lemma \ref{lemma_strong_normalization_and_additivity} below).

On the other hand, non-contextuality cannot be obtained from additivity alone either. By additivity alone, we mean additivity without ANC. It is difficult to separate the two, because some may claim that without ANC, additivity does not actually hold. However, it does hold in some unusual form. We first briefly explain why ANC is essential for the proof of ONC.

The logic of the proof in \cite{Logiurato_2012} can be illustrated with a simple example. Suppose we have two sets of projective operators $A_1,A_2,\dots$ and $B_1,B_2,\dots$. Let $A_1=B_1$, and the other $A_i$ and $B_i$ are generally different. We assume additivity and normalization, but not non-contextuality. This means that we may have

$$\mu(A_1)\ne \mu(B_1)$$

We can obtain strong normalization from normalization and additivity (see Lemma \ref{lemma_strong_normalization_and_additivity} below). We then have

$$1-\mu(A_1)=1-\mu(B_1)$$
$$\mu(A_2+A_3+...) \ne \mu(B_2+B_3+...)$$

The authors of \cite{Logiurato_2012} then claim that since the operators $A_2+A_3+\dots = B_2+B_3+\dots$, and they also share the same context $A_1=B_1$, even if ONC does not hold, we must still have $\mu(A_2+A_3+\dots) = \mu(B_2+B_3+\dots)$. This leads to a contradiction, and ONC is then derived from additivity. 

We now show why this proof is incorrect.

Consider a different setup. Let $A_1+A_2 = B_1+B_2$, but $A_i \ne B_j$ for any $i,j \in [2]$. The remaining $A_i$ and $B_i$ are generally different. Again, we assume additivity and normalization, but not non-contextuality. This means that we have

$$\mu(A_1+A_2)\ne \mu(B_1+B_2)$$
$$\mu(A_3+...) \ne \mu(B_3+...)$$

In this case, there is no contradiction, since the contexts of $A_3+\dots$ and $B_3+\dots$ are not the same.

Looking more closely at the setup discussed above, one may claim that additivity already implies

$$\mu(A_1+A_2)=\mu(B_1)+\mu(B_2)= \mu(B_1+B_2)$$

This is precisely where the confusion in the definitions arises. This equation holds only if ANC holds. ANC is a stronger non-contextuality assumption than ONC, and we can always derive the latter from the former. Although ANC is often implicitly assumed in additivity assumptions, it is important to recognize that non-contextuality is actually derived from a stronger version of non-contextuality, rather than from additivity. Without ANC, ONC cannot be derived from additivity alone.

A simple counterexample is the Equal rule. We assign $\mu(A)=1/N$ to every operator $A$ in any set of operators, where $N$ is the number of operators. Let the two sets of projective operators be $\{A_1,A_2,A_3\}$ and $\{B_1,B_2,B_3,B_4\}$, where $A_1+A_2 = B_1+B_2$ but $A_i \ne B_j$ for any $i,j \in [2]$. Then

$$\mu(A_1+A_2) = 2/3 = \mu(A_1) + \mu(A_2)$$
$$\mu(B_1+B_2) = 1/2 = \mu(B_1) + \mu(B_2)$$
$$\mu(A_1+A_2) \ne \mu(B_1+B_2)$$

This $\mu$ satisfies additivity without ANC or ONC.

\begin{remark}

This counterexample also applies to the case considered in \cite{Logiurato_2012}. Suppose we have two sets of projective operators $\{A_1,A_2\}$ and $\{B_1,B_2,B_3\}$, where $A_1=B_1$. Since both sets sum to $I$, we know that $A_2=B_2+B_3$. By the Equal rule, we obtain

$$\mu(A_2)=\frac{1}{2}\ne \mu(B_2+B_3)= \frac{2}{3}=\mu(B_2)+\mu(B_3)$$

It may seem confusing that we obtain this inequality when the operator $A_2=B_2+B_3$ and the context $A_1=B_1$ are both equivalent. The correct interpretation is that, for the measurement $\{B_1,B_2,B_3\}$, we can only calculate $\mu(B_2+B_3)$ using additivity. Otherwise, we have no way of assigning a value to $\mu(B_2+B_3)$. If we measure $B_2+B_3$ directly, the measurement operator changes. Therefore, the measurement operators are not the same in the two measurements, although the context is.

\end{remark}

Thus, we have proved that non-contextuality and additivity are independent properties. The ONC assumption cannot be derived from additivity without ANC, and the derivation in the opposite direction is even more clearly problematic. We will show below why assuming strong normalization can falsely suggest that additivity can be derived from non-contextuality, or from any property we want.

\subsection{Additivity and Normalization}
\label{subsection_additivity_and_normalization}

The relationship between additivity, normalization, and strong normalization is as follows.

\begin{lemma}
\label{lemma_strong_normalization_and_additivity}

$\mu$ is strongly normalized if and only if $\mu$ is countably additive and normalized.

Notice that if the set of operators is finite rather than countable in the definition of strong normalization, we can only prove that $\mu$ is finitely additive.

\begin{proof}

If $\mu$ is strongly normalized, we know that $\mu(I)=1$, and for mutually compatible $A_i$,

$$\mu(\sum_i A_i) + \mu(I-\sum_i A_i)= \sum_i \mu (A_i) + \mu(I-\sum_i A_i) = 1$$

Thus,

$$\mu(\sum_i A_i) = \sum_i \mu (A_i) $$

Conversely, if $\mu$ is countably additive, then for any countable set of mutually compatible operators $\{A_i\}$ such that $\sum_i A_i = I$,

$$\mu(I)=\mu(\sum_i A_i)=\sum_i\mu(A_i)$$

Since $\mu$ is normalized, we have $\sum_i \mu(A_i)=1$.

\end{proof}

\end{lemma}

Thus, if strong normalization is used as the normalization assumption, additivity is guaranteed. Since normalization is a very elementary and natural assumption, it is easy to mistakenly conclude that, with this seemingly harmless normalization assumption, additivity follows trivially from other assumptions such as non-contextuality. This conclusion is incorrect. In the derivation of the Born rule, additivity is a crucial assumption which is also indispensable (see more in Section \ref{section_additivity_in_5_proofs}).

We have already presented a counterexample in Section \ref{subsection_Additivity_and_Non-Contextuality} that satisfies normalization but not strong normalization or additivity. This implies that normalization and additivity are also independent properties and cannot be derived from each other.

\section{Additivity in Existing Derivations of the Born Rule}
\label{section_additivity_in_5_proofs}

We have proved in the previous section that additivity cannot be derived from non-contextuality or normalization. In this section, we will show the indispensability of the additivity assumption in five important existing derivations of the Born rule: two of them are Gleason's Theorem \cite{gleason1957_original_GT} and its variant \cite{busch2003}; the other three are derivations proposed by supporters of the Many-Worlds Interpretation (MWI) \cite{deutsch1999,zurek2005,hartle2019quantummechanicsindividualsystems}. Note that although there are clear features of MWI in these three proofs, the mathematical proofs do not rely on any special assumptions of MWI. Thus, we can regard all five derivations as independent of quantum interpretations.

We will give a brief introduction to each derivation, including its proof structure and postulates. This helps us identify the additivity assumptions that appear in the derivations, as well as the role that additivity plays in them. We do not discuss all the debates surrounding each derivation in detail, since doing so would involve too much content and would distract from the focus of this paper. Instead, we focus on issues related to additivity.

\subsection{Gleason's Theorem}
\label{subsection_GT}

Proved in 1957, Gleason's Theorem is one of the earliest works deriving the Born rule \cite{gleason1957_original_GT}, and is often regarded as one of the most important and reliable results in the field \cite{bell1966,Barnum_2000_Quantum_probability_from_decision_theory,Logiurato_2012}. It replaces the commonly criticized strong additivity assumption in von Neumann’s first derivation \cite{vn2018} with a more reasonable additivity assumption, and derives the Born rule with a very rigorous mathematical framework. The main concerns about Gleason's Theorem are that it does not have a very clear physical interpretation \cite{Schlosshauer_2005,Logiurato_2012,Zurek_2018,lawrence2023bornruleaxiom}, and it has an exception for two-dimensional Hilbert spaces. Although Gleason's original proof is often regarded as intricate \cite{Cooke_Keane_Moran_1985}, we can briefly rephrase its structure as follows.

First, we summarize the postulates assumed by Gleason. The only non-probabilistic quantum postulate used by Gleason is the Observable postulate (QM3). He also uses the Real Function postulate (M1). The function $\mu$ he defines takes a single projection $A_i$ as input, so the measurement context can be undetermined. The remaining postulates used by Gleason are listed below.

(G1). \textbf{Frame Function postulate (Additivity and Normalization)} --- Gleason assumes $\mu$ to be a frame function:

\begin{definition}

A frame function $\mu(\ket{x})$ of weight $W$ is a real-valued function defined on the set of unit vectors in a separable real or complex Hilbert space $\mc{H}$, such that if $\{\ket{i}\}_i$ forms an orthonormal basis of $\mathcal{H}$, then

$$
    \sum_i \mu(\ket{i})=W
$$

We also say “frame function on $\mc{H}$” below for simplicity.

\end{definition}

Note that this definition is equivalent to our strong normalization when $W=1$. Therefore, Gleason's frame function is a combined assumption of additivity and normalization. The normalization property is discussed separately below.

When $\dim{\mc{H}}<\aleph_0$, the additivity of the frame function is finite additivity; when $\dim{\mc{H}}=\aleph_0$, it is countable additivity.

(G2). \textbf{Non-Negativity postulate} --- For $\forall \ket{x}, ~~ \mu(\ket{x}) \ge 0$.

(G3). \textbf{Normalization postulate} --- From Gleason's additivity, we can define $\mu(I) = W$.

At first glance, the normalization assumption may seem not very useful. One might think that, given Gleason's additivity with ANC, $\mu(I)$ is guaranteed to be a fixed value $W$, and that $W$ can always be rescaled. However, from the definition of a measure, it is possible that $W=+\infty$. This possibility may lead to a failure in the proof of the continuity of $\mu$. Therefore, the normalization assumption and the non-negativity assumption are both introduced to exclude the infinite output values. 

Less importantly, if we want to interpret $\mu$ as a probability, we should assume normalization and rescale $W$ to 1 at the end to make the result physically meaningful and consistent with the Born rule.

(G4). \textbf{ONC postulate} --- $\mu$ depends only on its input vector, namely the single projective operator measured, and not on the measurement context. That is, for any two different contexts $\{x_i\}_i$ and $\{x'_i\}_i$, as long as $x_k=x'_k$, we have

$$\mu(\ket{x_k}|\{x_i\}_i) = \mu(\ket{x'_k}|\{x'_i\}_i)$$

The non-contextuality assumption is not stated clearly in Gleason's original paper \cite{gleason1957_original_GT}. However, in modern discussions of Gleason's Theorem, non-contextuality is often the first postulate to be mentioned and is regarded as very important \cite{Barnum_2000_Quantum_probability_from_decision_theory,Caves_2004_Gleason_Type_Derivations,Logiurato_2012,lawrence2023bornruleaxiom}. This emphasis may be influenced by the early contextual counterexample given by Bell, which satisfies Gleason's additivity and other postulates except non-contextuality, as well as by the connection of this counterexample with hidden variable interpretations of quantum mechanics \cite{bell1966}.

We also list two other definitions used in Gleason's proof to clarify some concepts in the discussion below. Note that these are not assumptions.

\begin{definition}

\cite{gleason1957_original_GT} Given a frame function $\mu$ on $\mc{H}$, we say it is regular if and only if there exists a Hermitian operator $\rho$ on $\mc{H}$ such that

$$
\mu(\ket{x}) = \bra{x} \rho \ket{x} 
$$

\end{definition}

\begin{definition}

\cite{gleason1957_original_GT} A completely real subspace $\mc{K}$ of a Hilbert space $\mc{H}$ is a linear subspace of $\mc{H}$ such that for any $\ket{x},\ket{y} \in \mc{K}$, $\langle x|y \rangle \in \mbb{R}$.

\end{definition}

It is easy to prove that a completely real subspace $\mc{K}$ is always a real Hilbert space. However, $\mc{H}$ itself does not need to be a real Hilbert space.

We now present the key ideas of Gleason's proof \cite{gleason1957_original_GT}. Note that the order here is not exactly the same as in Gleason's original paper, but is rearranged for clarity.

First, it is easy to show that a frame function $\mu$ on $\mc{H}$ remains a frame function when its domain is restricted to any subspace of $\mc{H}$. Moreover, if $\mu \geq 0$, it remains non-negative on any subspace.

Second, it is straightforward that any real or complex Hilbert space $\mc{H}$ with $\dim{\mc{H}} \ge 3$ always contains three-dimensional completely real subspaces $\mc{H}'\cong \mbb{R}^3$.

It is also straightforward that every completely real two-dimensional subspace of $\mc{H}$ is a linear subspace of some completely real three-dimensional subspace $\mc{H}' \cong \mbb{R}^3$. Furthermore, if a non-negative frame function $\mu$ is regular on $\mbb{R}^3$, then it is also regular on any completely real two-dimensional subspace of $\mbb{R}^3$.

Third, Gleason proves that if a non-negative frame function $\mu$ is regular on every completely real two-dimensional subspace of $\mc{H}$, then $\mu$ is regular on the whole space $\mc{H}$ (Lemma 3.3 and Lemma 3.4 in \cite{gleason1957_original_GT}).

Therefore, the crucial step is to prove that any non-negative frame function $\mu$ is regular on $\mbb{R}^3$. Gleason achieves this by first proving that any continuous frame function on $\mbb{R}^3$ is regular (Theorem 2.3 in \cite{gleason1957_original_GT}), and then proving that any non-negative frame function $\mu$ on $\mbb{R}^3$ is continuous (Theorem 2.8 in \cite{gleason1957_original_GT}).

The additivity of the frame function is essential in both steps. Clearly, a continuous function without additivity on $\mbb{R}^3$ is not necessarily regular, and a non-negative function without additivity on $\mbb{R}^3$ is not necessarily continuous.

Moreover, we cannot prove that an additive function on $\mbb{R}^3$ is continuous without the non-negativity assumption. The additivity of the frame function only guarantees linearity for rational coefficients, that is, $\mu(qA)=q\mu(A)$ for all $q \in \mbb{Q}$. However, over the real coefficients there exist many pathological additive functions that are discontinuous and unbounded. Gleason therefore introduces the non-negativity assumption to exclude these discontinuous frame functions \cite{gleason1957_original_GT} (see more in Appendix \ref{Appendix_Additive_Function_and_Continuity}).

Thus, both additivity and non-negativity are indispensable in Gleason's proof. Additivity plays the more important role as the central postulate. Without it, the entire proof would fail.

\subsection{Busch's Extension of Gleason's Theorem to POVMs}

Busch extends Gleason's result in a 2003 paper to more general quantum measurements, represented by POVMs \cite{busch2003}. The advantage of Busch's proof is that it is also valid for the 2-dimensional Hilbert space, where Gleason's theorem does not single out the Born rule as the only solution. Busch's proof is also simpler, clearer, and more intuitive. The price, however, is an additivity assumption much stronger than the one used by Gleason. Note that Busch's additivity is stronger because it concerns POVM effects. The probabilistic result generated from assumptions about POVM effects can be suspicious, since they do not correspond to physical values of isolated systems as projective operators \cite{Grudka_2008}. Nevertheless, the additivity applies only to compatible POVM effects, so the criticism \cite{Hermann2016_Determinism_and_Quantum_Mechanics,bell1966} of von Neumann's strong additivity assumption (which applies to any measurements, including incompatible ones) does not apply to Busch's additivity assumption. Both Gleason and Busch make an explicit additivity assumption in their proofs, so we focus on analyzing the role of additivity in their arguments in the previous and current subsections. In the other three derivations, the additivity assumption is not always presented clearly in the original works. We will analyze problems that this omission can cause.

We first briefly introduce Busch's postulates and the structure of his proof.

Busch assumes the POVM postulate (QM3') and the Real Function postulate (M1). For (M1), Busch actually assumes $\mu$ to be a probability measure, but we separate the properties of normalization and non-negativity. The remaining postulates by Busch are listed below:

(G-B1). \textbf{Countable Additivity postulate for POVM} --- $\mu$ satisfy the additivity property

$$
\mu(A+B+...)
= \mu(A) +\mu(B) + ...$$ 

if $|\{A,B,...\}| \leq \aleph_0$ and $A+B+...\le I$, i.e., the POVM effects are compatible.
 
(G-B2). \textbf{Non-Negativity postulate} --- $\mu \geq 0$.

In Gleason's proof, non-negativity is introduced to ensure continuity. It plays the same role in Busch's proof. Busch does not prove this continuity in his original paper. We provide a proof in Theorem \ref{thm_continuity_by_busch} in Appendix \ref{Appendix_Additive_Function_and_Continuity}.

(G-B3). \textbf{Normalization postulate} --- $\mu(I)=1$. 

As in Gleason's proof, this assumption is useful for ruling out infinite output values and for aligning the result with the Born rule.

(G-B4). \textbf{ONC postulate} — For any POVM effect $A$, $\mu(A)$ depends only on $A$ and not on other compatible measurements performed simultaneously.

We present another definition of linear functional here to clarify Busch's proof. This is not an assumption.

\begin{definition}
\label{def_linear_functional}

\cite{Axler2024_chap3} For a vector space $V$ over the field $K$, $f:V \rightarrow K$ is called a linear functional if for any $\vec{u},\vec{v} \in V$ and any scalar $c \in K$,

    (i)  $f(\vec{u} +\vec{v} )=f(\vec{u} )+f(\vec{v} )$
    
    (ii)  $f(c\vec{u} )=cf(\vec{u} )$
        
\end{definition}

Busch's proof is structured as follows \cite{busch2003}. It is a proven result that any linear functional $\mu(A)$ can be written as $\mu(A) = \tr{[\rho A]}$ \cite{vn2018}. Since Busch assumes countable additivity for POVM effects at the outset, what remains is to derive homogeneity, i.e. $\mu(rA)=r \mu(A)$ for $r\in \mbb{R}$. The proof for $r\in \mbb{Q}$ is straightforward. Then continuity is needed to extend homogeneity to real $r$'s. 

Although Busch's additivity is much stronger than Gleason's, establishing the required continuity (which differs from the continuity discussed in Gleason's proof; see Remark \ref{remark_diff_continuity} in Appendix \ref{Appendix_Additive_Function_and_Continuity}) still relies on the normalization and non-negativity assumptions. Nevertheless, additivity for POVM effects simplifies the proof substantially and also plays a central role at each step of Busch's derivation.

\subsection{Deutsch's Derivation with Decision Theory}

Deutsch gave another derivation of the Born rule in 1999 with the help of decision theory \cite{deutsch1999}. This is a significant result in the field. It aims to derive the probabilistic Born rule from non-probabilistic postulates from quantum mechanics and decision theory, which Deutsch describes as “Deriving a ‘tends to’ from a ‘does’.” \cite{deutsch1999} His method was later refined by Wallace and is thus also called the Deutsch-Wallace Theorem \cite{wallace2009}. Discussions about whether it can effectively achieve its goal have been ongoing \cite{zurek2005,Vaidman2012,price2008decisionsdecisionsdecisionssavage,DAWID201455}. For our purposes, we mainly discuss two problems in Deutsch's derivation: the continuity problem pointed out by Barnum \cite{barnum2003nosignallingbasedversionzureksderivation}, and the dimensionality problem pointed out by Caves \cite{Caves2005NotesOnZurek}. Note that both criticisms were originally raised in the context of Zurek's derivation (see Section \ref{subsection_Zurek_proof}) but also apply to Deutsch's derivation. Both problems are closely related to the additivity assumption. In short, although Deutsch’s additivity assumption appears very different from countable additivity at first glance, countable additivity can in fact be derived from Deutsch’s postulates. Together with another hidden non-negativity assumption in Deutsch’s derivation, the continuity of the function $\mu$ is actually guaranteed, and thus continuity is not a genuine problem. However, since Deutsch’s additivity concerns only the whole set of projective operators, there is no ANC property in Deutsch’s hidden additivity. Moreover, he uses a hidden ONC assumption not from the beginning of the proof, but only at a later step for states with unequal amplitudes. This leaves room for contextual non-Born-rule counterexamples in low-dimensional Hilbert spaces. To analyze these problems more clearly, we first briefly go through Deutsch’s postulates and proof.

Deutsch starts his proof from a value function $\mc{V}[\ket{\psi},X]$, where $X$ is an observable, and proves that this value function yields the expectation value calculated by the Born rule, if we regard quantum measurement as a game played by a rational player following certain decision theory rules. In this game, the reward of each game is equal to the measurement outcome, namely a certain eigenvalue. In this work, for consistency of comparison, we can set the eigenvalue of each eigenstate $\ket{x_i}$ to be 1. Then the expectation value for any outcome is simply the probability, and we can write $\mu[\ket{\psi},X]=\mc{V}[\ket{\psi},X]$.

Although the rationality used by Deutsch is questioned in the literature \cite{price2008decisionsdecisionsdecisionssavage}, the rules that the player must follow in this quantum measurement game are given clearly. We summarize all of Deutsch’s assumptions as follows.

First, the non-probabilistic quantum postulates used by Deutsch are the State postulate (QM1), the Observable postulate (QM3), and the Eigenstate postulate (QM4). The Real Function postulate (M1) is also assumed in the value function setup. The remainder of his explicit and implicit postulates are listed below:

(D1). \textbf{Deutsch’s Additivity postulate} — The player is indifferent between receiving a single reward $(a + b)$ in one game and receiving separate rewards $a$ and $b$ in two games.

This additivity assumption appears trivial and reasonable. However, Deutsch actually uses it to justify another important additivity, or strong normalization,

\begin{equation}
\label{Eqn_Deutsch's_strong_noramlization}
\sum_i \mu(\ket{\psi},\ket{x_i}\bra{x_i})= 1
\end{equation}

The original value function $\mc{V}$ proposed by Deutsch does not take a single projective operator $\ket{x_k}\bra{x_k}$ as input. However, we can easily define it by considering a game in which we only receive the reward $x_k$ when the corresponding eigenstate is measured, and all other $x_i$’s are discarded.

It is from this additivity assumption that the whole proof starts. It is equivalent to Deutsch’s Additivity postulate if we consider the situation of a game in real life --- It is trivial that, in a game, the sum of probabilities of all possible events should be 1. However, here Deutsch already introduces a probabilistic assumption. We will discuss in detail below how Equation (\ref{Eqn_Deutsch's_strong_noramlization}) is implicitly introduced and used.

(D2). \textbf{Non-negativity postulate} — $\mu(\ket{\psi},\ket{x_i}\bra{x_i}) \ge 0$.

This assumption is not stated explicitly by Deutsch. However, when he discusses the result for real probabilities, he effectively invokes it by claiming that if the eigenvalue for a certain eigenstate is increased while the remaining eigenvalues are unchanged, the player will prefer the new game to the old one.

Although Deutsch’s proof for the result for any real probability is problematic, this assumption does guarantee the continuity of the function $\mu$ and thus the Born rule for real probabilities. This is because it implies the non-negativity of the probability: if the probability $\mu$ of obtaining a particular eigenstate were negative, then the player should not want the corresponding eigenvalue to be larger.

(D3). \textbf{Transitivity postulate} — If the player prefers reward $a > b$ and $b > c$, then he also prefers $a > c$.

(D4). \textbf{Substitutability postulate} — In a composite game, the player is indifferent if any sub-game is replaced by another game with the same value.

This assumption is not very useful, although it appears essential at the step where the result for unequal probabilities is proven. However, the real key point is to determine when two games have the same value. This is precisely where Deutsch’s hidden non-contextuality assumption enters.

(D5). \textbf{Zero-Sum Rule postulate} — If the player always receives opposite rewards $(x$ and $-x)$ for two games, the values of the two games sum to 0.

(D6). \textbf{State Non-Contextuality and ONC postulate} --- $\mu(\ket{\psi},\ket{x_i}\bra{x_i})$ only depends on the state $\ket{\psi}$ and the projective operator $\ket{x_i}\bra{x_i}$.

This hidden ONC assumption is pointed out by Logiurato and Smerzi \cite{Logiurato_2012}. It is only essential at the step of proving the result for unequal probabilities, where compatible measurements on an ancillary system are used. However, that step also requires infinite dimensionality. Therefore, for finite-dimensional Hilbert spaces, there is no ONC assumption in Deutsch's proof, and contextual counterexamples exist (see more below, especially Equation (\ref{Eqn_counterexample_for_finite_space_in_Deutsch_proof})).

The structure of Deutsch's proof is as follows \cite{deutsch1999}. He first proves that for the simplest 2-dimensional symmetric pure state $\ket{\psi}=\frac{1}{\sqrt{2}}(\ket{x_1}+\ket{x_2})$, $\mu(\ket{\psi},\ket{x_i}\bra{x_i})$ must be equal to the modulus squared of the amplitude of the eigenstate $\ket{x_i}$. He then extends this result to the general equal-amplitude pure state $\ket{\psi}=\frac{1}{\sqrt{2}}(\ket{x_1}+\ket{x_2}+...+\ket{x_n})$ for any positive integer $n$. Next, for pure states with unequal amplitudes, Deutsch starts from states with rational amplitudes,

$$\ket{\psi}=\sqrt{\frac{m}{n}}\ket{x_1}+\sqrt{\frac{n-m}{n}}\ket{x_2}$$

Using his additivity, Deutsch proves that the value for this game is equal to another game $\mc{V}[\ket{\psi_{XY}},X\otimes Y]$, where

\begin{equation}
\label{Eqn_Deutsch_state_XY}
\ket{\psi_{XY}}=\sum_{j=1}^{m}\sqrt{\frac{1}{n}}\ket{x_1}\ket{y_j}+\sum_{j=m+1}^{n}\sqrt{\frac{1}{n}}\ket{x_2}\ket{y_j}
\end{equation}

and 

$$\sum_{j=1}^m \ket{y_j} = \sum_{j=m+1}^n \ket{y_j} =0$$

The result then follows from the equal-amplitude case. It is then straightforward to extend the result to any $\ket{\psi} = \sum \sqrt{q_i} \ket{x_i}$ with $q_i \in \mbb{Q}$.

The extension to real probabilities proposed by Deutsch is problematic. He claims that we can always find unitary transformations $U$ that transform each $\ket{x_i}$ into an eigenstate with a higher eigenvalue, and that the player will prefer this new game to the old one. Deutsch also notes that transforming the eigenvalue is not relevant to the amplitude. He therefore introduces a more complicated operation that transforms each $\ket{x_i}$ into a superposition of itself and other eigenstates, all with higher eigenvalues. He claims that it is possible to find a $U$, such that the state after the transformation, $U\ket{\psi}$, has amplitude $q_i$ with $|q_i|^2 \in \mbb{Q}$ for each eigenstate of $X$. Then the value of this new game becomes an upper bound for the initial game. Since rational numbers are dense in the reals, he argues that we can prove the result for any amplitude $\sqrt{r_i} ~~ (r_i \in \mbb{R})$ by squeezing it between upper and lower bounds obtained from limits of the $\sqrt{q_i}$’s.

However, even if this complex operation is possible, obtaining upper or lower bounds for each amplitude $\lambda_i$ in the original state $\ket{\psi}$ means making each amplitude $\lambda_i$ smaller or larger, which contradicts normalization condition for a quantum state. 

Nevertheless, this step of the proof can still be made to work if we can prove that $\mu$ is continuous with respect to the amplitudes, as we will show below. 

Deutsch then extends the result to the case where $\lambda_i$ is any complex number by arguing that applying unitary transformations that change only certain partial phase before and after the measurement will give the player equal rewards. This ends his proof \cite{deutsch1999}.

We now analyze the problem of continuity raised by Barnum \cite{barnum2003nosignallingbasedversionzureksderivation}. We have shown that Deutsch's proof of continuity does not work. However, since countable additivity and non-negativity are both implicitly present in Deutsch's derivation, the continuity of $\mu$ is in fact guaranteed by Gleason's Theorem. We first show where Deutsch introduces the hidden assumption of countable additivity.

Deutsch claims that the following equation follows from Deutsch's Additivity:

\begin{equation}
\label{Eqn_Deutsch's_+k_additivity}
\mc{V}\left[ \sum_i \lambda_i \left\lvert x_i + k \right\rangle \right] = k + \mc{V} \left[\sum_i \lambda_i \left\lvert x_i \right\rangle \right] 
\end{equation}

The idea is simply that if we add $k$ to each reward, the final expected reward also increases by $k$. From our definition of $\mu(\ket{\psi},\ket{x_i}\bra{x_i})$, we have

$$\mc{V}\left[ \sum_i \lambda_i \left\lvert x_i \right\rangle \right] = \sum_i \mu(\ket{\psi},\ket{x_i}\bra{x_i}) x_i $$

and Equation (\ref{Eqn_Deutsch's_+k_additivity}) is only valid when

$$\sum_i \mu(\ket{\psi},\ket{x_i}\bra{x_i}) = 1$$

This is the truly important additivity assumption in Deutsch's proof, and it is exactly Gleason's additivity, or strong normalization, except without ANC.

Although we could conclude that function $\mu$ must be continuous from Gleason's Theorem, we cannot say that this continuity can be proved in the same way as in Gleason's proof. Gleason's notion of continuity concerns eigenvectors rotating by a small angle \cite{gleason1957_original_GT}, rather than small changes in amplitudes. These differences between the various forms of continuity are easy to confuse and can lead to misunderstandings (see more in Appendix \ref{Appendix_Additive_Function_and_Continuity}, especially Remark \ref{remark_diff_continuity}).

At last, we expand the discussion by Caves concerning dimensionality \cite{Caves2005NotesOnZurek}. This issue is similar to the exception for $\dim{(\mc{H})}=2$ in Gleason's Theorem, but Deutsch's derivation leaves space for more exceptions.

Taking a closer look at the proof, we notice that only the first step is valid for a genuine 2-dimensional Hilbert space --- this means that the 2-dimensional space is not a subspace contained in a larger Hilbert space. The second step, for $\ket{\psi}=\frac{1}{\sqrt{2}}(\ket{x_1}+\ket{x_2}+...+\ket{x_n})$, uses induction, and for odd $n$ we need to derive the result from an $(n+1)$-dimensional space. This means, for example, that we need a 4-dimensional space to prove the Born rule for $\ket{\psi}=\frac{1}{\sqrt{2}}(\ket{x_1}+\ket{x_2}+\ket{x_3})$. To derive the Born rule for any rational $m/n$, we need an auxiliary system of arbitrary integer dimension. This can only be guaranteed when the auxiliary system has countably infinite dimension, which indicates $\dim(\mc{H}) = \aleph_0$. A $\aleph_0$-dimensional space is sufficient for the rest of the proof. Thus, Deutsch's proof actually shows that, for a state in any subspace of a $\aleph_0$-dimensional Hilbert space, the Born rule is the only solution for the value function $\mc{V}$.

We present a counterexample for $d=2$ here:

\begin{equation}
\label{Eqn_counterexample_for_finite_space_in_Deutsch_proof}
\mu(\ket{\psi},\ket{x_i}\bra{x_i})= \frac{ |\langle \psi|x_i\rangle|^4 }{ \sum_{j=1}^2 |\langle \psi|x_j\rangle|^4}   
\end{equation}

This counterexample satisfies all of Deutsch's postulates, but it apparently generates results different from the Born rule. Notice that when we have a 2-dimensional symmetric pure state, this $\mu$ function gives the same result as the Born rule --- namely,
$\mu(\frac{1}{\sqrt{2}}(\ket{x_1}+\ket{x_2}),\ket{x_1}\bra{x_1}) = \mu(\frac{1}{\sqrt{2}}(\ket{x_1}+\ket{x_2}),\ket{x_2}\bra{x_2}) = 1/2$.
This is the only result that can be derived for a 2-dimensional space from Deutsch's proof. For $d>2$, counterexamples of the same form are not valid counterexamples to Deutsch's proof, because higher dimensionality imposes more constraints. However, as long as $d<\aleph_0$, there could be states $\ket{\psi}$ for which $\mu(\ket{\psi},X)$ do not obeys the Born rule, and other counterexamples can still be constructed.

We should note that this counterexample still satisfies both Deutsch's Additivity postulate and strong normalization without ANC, but it does not satisfy non-contextuality. This is reasonable, because Deutsch introduces an ONC postulate only implicitly when proving the result for unequal amplitudes \cite{Logiurato_2012}. If ONC does not hold, we can have

$$\mc{V}(\sqrt{\frac{m}{n}}\ket{x_1}+\sqrt{\frac{n-m}{n}}\ket{x_2},X) \ne \mc{V}(\ket{\psi_{XY}},X\otimes Y)$$

where $\ket{\psi}_{XY}$ is as in Equation (\ref{Eqn_Deutsch_state_XY}). In that case, the proof does not work. The ONC assumption is introduced only at this step, but the infinite dimensionality of the auxiliary system $Y$ is also introduced at the same step. Therefore, counterexamples exist for any finite-dimensional Hilbert space.

\subsection{Zurek's Derivation with Envariance}
\label{subsection_Zurek_proof}

People are unsatisfied with Gleason's Theorem because of its lack of physical intuition \cite{Schlosshauer_2005,Logiurato_2012,Zurek_2018,lawrence2023bornruleaxiom}, and Deutsch's appeal to decision theory and rationality also seems suspicious to many \cite{zurek2005,Vaidman2012,price2008decisionsdecisionsdecisionssavage,DAWID201455}. By contrast, Zurek aims to provide a more physically motivated derivation of the Born rule \cite{Zurek_2003,zurek2005,Zurek_2018}. The structure of Zurek's proof is actually very similar to Deutsch's, except for the first few steps. Zurek proves the Born rule for the equal-amplitude qubit state using a property he calls envariance, which was first abbreviated as “environment-assisted invariance” in his 2003 paper \cite{Zurek_2003}, but was later updated to the more appropriate term “entanglement-assisted invariance” \cite{Zurek_2018}. This grounds Zurek's proof in the physically meaningful property of quantum entanglement.

Zurek gives special attention to the additivity assumption in his works. He claims that Gleason's additivity can be derived from a ``weak" additivity assumption \cite{zurek2005}. However, the consequence is that Zurek's additivity is in fact weaker than Gleason's additivity or Deutsch's implicit countable additivity, and thus cannot guarantee the continuity of the measurement function $\mu$. Due to the similarity between Zurek's and Deutsch's derivations, it also suffers from the same problem of dimensionality. We will first briefly introduce Zurek's proof and postulates, and then explain these two problems.

There is a very interesting innovation in Zurek's assumptions. He does not assume the orthogonality of eigenstates or the expression of observables from the Observable postulate (QM3), but instead derives them from the State postulate (QM1), the Evolution postulate (QM2), and the Eigenstate postulate (QM4) \cite{Zurek_2018}. In order not to distract from our main focus, we rephrase this proof in Appendix \ref{Appendix_Zurek_proof_QM3}. Nevertheless, this is a very inspiring result worth noting.

Zurek also uses another not-so-common non-probabilistic quantum postulate: the representation of a composite state is a vector in the tensor product space of the components' Hilbert spaces \cite{Zurek_2018}. The term ``probability" is used directly by Zurek to describe the measurement result, so we can conclude that the Real Function postulate (M1) is also assumed. Since the observable $X$ or the Schmidt decomposition of the state is given at the beginning, we can also represent the measurement function used by Zurek as $\mu(\ket{\psi},X)$. However, we do not need to assume all properties of a probability measure here, because the additivity and normalization are assumed independently by Zurek, and the non-negativity is not useful in the proof.

The additional postulates beyond non-probabilistic quantum mechanics used by Zurek are:

(E1). \textbf{Zurek's Weak Additivity and Normalization postulate} --- For two mutually exclusive events $\kappa$ and $\kappa^\perp$,

$$\mu(\kappa \cup \kappa^\perp) = 1=\mu(\kappa)+\mu(\kappa^\perp)$$

This means that the event $\kappa$ either occurs or does not occur. We also use measurement outcomes $x_i$ to represent events for simplification.

Zurek proves that the usual countable additivity can be derived from this weak additivity \cite{zurek2005}. However, we show that it ensures additivity only for states having rational probabilities for all outcomes (see proof in Appendix \ref{Appendix_Zurek's_Additivity}). That is, if outcome events $\{x_i\}_i$ are mutually disjoint, and each $\mu(\ket{x_i}\bra{x_i}) \in \mbb{Q}$, then

$$\mu \left (\sum_i \ket{x_i}\bra{x_i} \right) = \sum_i \mu \left (\ket{x_i}\bra{x_i} \right) $$

Therefore, Zurek’s additivity is not equivalent to Gleason’s additivity or Deutsch’s hidden countable additivity, but is weaker than both. The continuity of $\mu$ cannot be guaranteed even with non-negativity.

(E2). \textbf{Zurek's Locality postulate} --- The state of system $S$ cannot be affected by any operation on the environment $E$ alone, and vice versa.

(E3). \textbf{State Non-Contextuality postulate} --- For any observable, the measurement results fully depend on the quantum state $\ket{\psi}$. Moreover, if the measured system is a subsystem of a larger composite system, the composite state determines the state of the subsystem.

This corresponds to “fact 2” and “fact 3” in Zurek’s paper \cite{Zurek_2018}. Note that this assumption appears to be a form of non-contextuality, but it is not the same as the usual ONC and cannot play the role of ONC. We refer to this assumption as State Non-Contextuality. Deutsch also adopts this State Non-Contextuality postulate \cite{deutsch1999}. The main use of this postulate is to rule out any possible influence from hidden variables. Gleason’s and Busch’s approaches do not include this type of non-contextuality, because they do not assume the State postulate (QM1) at the beginning \cite{gleason1957_original_GT,busch2003}.

(E4). \textbf{ONC postulate} --- For any $\ket{\psi}$, $\mu(\ket{\psi},\ket{x_i}\bra{x_i})$ depends only on the projective operator $\ket{x_i}\bra{x_i}$, regardless of the whole observable.

The proof for states with unequal amplitudes given by Zurek is exactly the same as that of Deutsch. Thus, both of them implicitly introduce the ONC postulate at that step.

The most important concept in Zurek’s derivation, envariance, is not a postulate but a symmetry property of an entangled state under certain unitary transformations. 

\begin{definition}

\cite{zurek2005} We call a state $\ket{\Psi_{SE}}$ envariant under a unitary transformation $U_S$ if

$$
(I\otimes U_E)(U_S \otimes I) \ket{\Psi_{SE}} = \ket{\Psi_{SE}}
$$

This means that a unitary transformation $U_S$ acting only on the measured system can be undone by another unitary transformation acting only on the environment.

\end{definition}

Zurek’s proof then follows this structure \cite{zurek2005}: first, he proves the Born rule for equal-amplitude states using the envariance property. Consider an equal-amplitude state

$$\ket{\Psi_{SE}} = \frac{1}{\sqrt{2}}(\ket{x_1}\ket{E_1}+\ket{x_2}\ket{E_2})$$

We can swap the eigenstates of the measured system using the unitary operator

$$U_S \otimes I= (\ket{x_1}\bra{x_2}+\ket{x_2}\bra{x_1}) \otimes I$$

Let the initial probabilities for outcomes $x_1$ and $x_2$ be $p_1$ and $p_2$. Before the swap, when the environment is in eigenstate $\ket{E_1}$, the corresponding eigenstate for the system is $\ket{x_1}$; after the swap, it becomes $\ket{x_2}$. Thus, for the state after the first swap, we have $\mu(x_1) = \mu(E_2) = p_2$ and $\mu(x_2) = \mu(E_1) = p_1$. We then perform another swap acting only on the environment,

$$I \otimes U_E = I \otimes (\ket{E_1}\bra{E_2}+\ket{E_2}\bra{E_1})$$

The final composite state is then the same as the initial one, and $\mu(x_1) = p_1$ and $\mu(x_2) = p_2$ for the final state after the two swaps. However, from Zurek’s Locality postulate, an operation on the environment $E$ alone should not change the probability distribution of the system $S$. Therefore, we must conclude that $\mu(x_1) = p_2 = p_1$. From Zurek’s Weak Additivity and Normalization postulate, it follows that $p_1 = p_2 = 1/2$. 

This result can be extended to any equal-amplitude state with $n$ outcomes, yielding $p(x_i) = 1/n$, by performing the swap multiple times.

The remainder of Zurek’s proof parallels that of Deutsch \cite{zurek2005,Zurek_2018}. Similarly, extending the result to real probabilities requires continuity of $\mu$. However, unlike in Deutsch’s proof, this continuity cannot be guaranteed by Zurek’s postulates. We analyze this problem below.

Zurek discusses the continuity problem explicitly in his paper \cite{zurek2005}. His proof is as follows: if event $\kappa_1 \subseteq \kappa_2$, then the probability must satisfy $\mu(\kappa_1) \leq \mu(\kappa_2)$. Then the real probability of event $\kappa_1$ can be approached by events $\kappa_2$ with rational probabilities from above. Continuity then follows by taking limits. 

This statement is very confusing. In quantum measurements, events correspond to measurement outcomes or eigenvalues of an observable. It is unclear what an inclusion relation between two different eigenvalues or eigenspaces would mean. We could possibly let $\kappa_2$ correspond to more than one eigenstate, including the one corresponding to $\kappa_1$. However, in that case the notion of “approaching” is not very meaningful. For example, let the quantum state be

$$\ket{\psi} = \sqrt{p_1} \ket{x_1} +  \sqrt{p_2} \ket{x_2} + ...$$

where $p_1 \notin \mbb{Q}$ and $p_1+p_2 \in \mbb{Q}$. Let $\mu(\kappa_1)= \mu(x_1)$ and $\mu(\kappa_2)= \mu(x_1 \vee x_2)$. Approaching $\mu(\kappa_1) = p_1$ by decreasing $p_2$ cannot prove the continuity of $\mu$ with respect to $\ket{x_1}$. Zurek’s mistake is similar to that of Deutsch. Deutsch attempts to take limits by increasing eigenvalues, while Zurek attempts to take limits by expanding the eigenspace. Neither approach satisfies our objective of proving continuity with respect to the amplitude or the probability.

Unlike Deutsch’s proof, the continuity of $\mu$ becomes a genuine problem here. This issue has been pointed out in the literature \cite{barnum2003nosignallingbasedversionzureksderivation}, but we present a clearer explanation of why continuity cannot be derived from Zurek's assumptions and proof. 

Although Zurek uses the term probability several times in his proof, and probabilities are generally expected to be non-negative, non-negativity is not required at any step of his proof. Moreover, even if non-negativity is assumed, Zurek’s additivity is still too weak to guarantee continuity. For example, we can let $\mu(x_i)$ follow the Born rule when the square of the amplitude $|\lambda_i|^2\in \mbb{Q}$, and let $\mu(x_i) = \sqrt{2}$ (or any irrational number) when $|\lambda_i|^2 \notin \mbb{Q}$. This counterexample satisfies Zurek’s additivity and non-negativity, but not the usual countable additivity (see more in Appendix \ref{Appendix_Zurek's_Additivity}), and is clearly discontinuous.

Zurek’s proof also lacks the ONC postulate until the step dealing with unequal amplitudes, since all eigenstates of the observable are fixed at the beginning when the entangled state is specified. Therefore, it suffers from the same dimensionality problem as Deutsch’s proof, namely that the result is only valid for $\aleph_0$-dimensional Hilbert spaces.

\subsection{Hartle's Frequentist Derivation}

Another early and important derivation of the Born rule was given by Hartle in 1967 using a frequentist method based on the infinite ensemble \cite{hartle2019quantummechanicsindividualsystems}. The same method was independently proposed earlier by Finkelstein in 1963 \cite{Finkelstein1963}, and it is therefore also known as the Finkelstein-Hartle Theorem \cite{Caves_2005_Properties_of_the_frequency_operator}. This derivation is rather controversial due to its use of infinity \cite{Farhi1989,SQUIRES199067,Caves_2005_Properties_of_the_frequency_operator,Logiurato_2012,vaidman2020_Derivations_Born_Rule}. More seriously, it contains some mathematical self-incoherence. For example, the infinite ensemble state $\ket{\psi}^{\otimes \aleph_0}$ appears to be both an eigenstate of the frequency operator $f^k$ and orthogonal to all initial eigenstates of $f^k$ \cite{SQUIRES199067,Caves_2005_Properties_of_the_frequency_operator}. We will point out a new incoherence, which can be resolved by adding an additivity assumption that is lacking in the original derivation. First, we review the postulates and the proof given by Hartle \cite{hartle2019quantummechanicsindividualsystems}.

Hartle assumes only a small number of postulates. He adopts the State postulate (QM1) and the Observable postulate (QM3) from non-probabilistic quantum mechanics. The Eigenstate postulate (QM4), however, is not assumed, because $\ket{\psi}^{\otimes \aleph_0}$ is generally not one of the initial eigenstates of $f^k$. Hartle describes the measurement result as a frequency, so we can think of the Real Function postulate (M1) as being assumed as well. We can denote the frequency function as $\mu(\ket{\psi},\ket{x_k}\bra{x_k})$. Normalization and non-negativity can be deduced trivially from Hartle’s description, but they are not useful for the proof. Apart from these, Hartle uses only one additional non-contextuality postulate:

(F1). \textbf{State Non-Contextuality and ONC postulate} --- The frequency of obtaining outcome $x_i$ in a quantum measurement depends only on the measured state $\ket{\psi}$ and the projective operator $\ket{x_i}\bra{x_i}$.

The structure of Hartle’s proof is as follows \cite{hartle2019quantummechanicsindividualsystems}. He aims to show that the frequency of obtaining an outcome $x_k$ in infinitely many measurements of the pure state $\ket{\psi}$ is equal to the probability given by the Born rule, $|\langle \psi | x_k \rangle|^2$. He defines a frequency operator $f^k_N$, which measures the number of times the outcome $x_k$ occurs in $N$ measurements:

\begin{align*}
&f_N^k = \sum_{i_1 = 1}^d \sum_{i_2 = 1}^d  ~ ... \\
&~ \sum_{i_N = 1}^d   \left( \frac{1}{N} \sum_{\alpha=1}^N \delta_{k i_\alpha} \right)  \ket{x_{i_1}}_{1} ...\ket{x_{i_N}}_{N} ~ \bra{x_{i_N}}_{N} ...\bra{x_{i_1}}_{1} 
\end{align*}

This operator $f^k_N$ is a composition of projections onto $N$ eigenspaces with degenerate eigenvalues ranging from $1/N$ to $1$ \cite{hartle2019quantummechanicsindividualsystems,Caves_2005_Properties_of_the_frequency_operator}. Hartle then shows that, in the limit $N \rightarrow \aleph_0$,

$$
\| f^k \ket{\psi}^{\otimes \aleph_0}- |\langle x_k|\psi \rangle|^2 \ket{\psi}^{\otimes \aleph_0} \| = 0
$$

Certainly, some special treatment is required to define $f^k$ as the limit of $f^k_N$ when $N \rightarrow \aleph_0$. We omit the discussion of this step here, since it is not directly relevant to our focus. Hartle’s treatment can be found in \cite{hartle2019quantummechanicsindividualsystems}, and criticisms by others can be found in \cite{Farhi1989,Caves_2005_Properties_of_the_frequency_operator}.

We want to point out a new self-incoherence in Hartle's proof that has not been discussed in the literature: Hartle bases his proof on the conjecture that the infinite ensemble $\ket{\psi}^{\otimes \aleph_0}$ of any pure state $\ket{\psi}$ is an eigenstate of the frequency operator $f^k$, and he proves this conjecture at the end \cite{hartle2019quantummechanicsindividualsystems}. If this is true, it seems trivial to conclude that the infinite ensemble $\rho^{\otimes \aleph_0}$ of any mixed state $\rho$ should also be an eigenstate of $f^k$ --- strictly speaking, $\rho^{\otimes \aleph_0}$ is the projector onto an eigenstate. We refer to it simply as an eigenstate for simplicity. 

We now check whether this can be true. Let a mixed state be

$$
\rho = \sum_j p_j \ket{\psi_j}\bra{\psi_j}
$$

In his proof, Hartle sets $f^k \ket{S} = 0$ if $\ket{S} \notin \mc{J}$, where $\mathcal{J} = \overline{\text{Span}} \left\{ \ket{\psi}^{\otimes \aleph_0} \ \Big| \ |\psi\rangle \in \mathcal{H}, \ \langle \psi|\psi\rangle = 1 \right\}$ \cite{hartle2019quantummechanicsindividualsystems}.  This indicates that he intentionally avoids discussing mixed states. From this definition, we have

$$
f^k\rho^{\otimes \aleph_0} = f^k \sum_j p_j^{\aleph_0} (\ket{\psi_j}\bra{\psi_j})^{\otimes \aleph_0} + 0 + 0 + ...  
$$

Therefore, $f^k\rho^{\otimes \aleph_0} = 0$ unless $\exists p_j = 1$, which means that $\rho$ is a pure state. On the contrary, if we calculate the result for a mixed state by applying the result for pure states, we obtain

$$
f^k \sum_j p_j (\ket{\psi_j}\bra{\psi_j})^{\otimes \aleph_0} = \sum_j p_j |\langle a_k | \psi_j \rangle|^2 (\ket{\psi_j}\bra{\psi_j})^{\otimes \aleph_0}
$$

Thus, these two calculations, both consistent with Hartle's key idea, lead to different results.

As mentioned above, Hartle avoids discussing mixed states in his proof. From his analysis, we can conclude that he does not regard $\rho^{\otimes \aleph_0}$ as an eigenstate of $f^k$. However, this is inconsistent with the frequentist conjecture of the proof. Moreover, the reason why the infinite state corresponding to $\rho$ should be $\sum_j p_j (\ket{\psi_j}\bra{\psi_j})^{\otimes \aleph_0}$ rather than $\rho^{\otimes \aleph_0}$ is not explained. Clearly, these two infinite states describe different experimental setups.

We can resolve this problem by adding a countable additivity assumption to Hartle's proof. We keep Hartle’s proof for pure states and assume that the frequency function $\mu(\rho,\ket{x_k}\bra{x_k})$ is countably additive with respect to disjoint pure states with any non-negative coefficient $p_j\rho_j = p_j \ket{\psi_j}\bra{\psi_j}$. Although additivity is usually assumed for compatible observables, it can also be applied to disjoint states, since the roles of the two are interchangeable in experiments --- we can always say that we are measuring some property of the measurement device using a given quantum state, rather than saying that we are measuring the state with the device. With this additivity assumption, infinite ensembles of mixed states do not need to be eigenstates of $f^k$, and the contradiction is removed.

\section{Conclusion}

In this work, we prove the indispensable role of additivity as an additional assumption in deriving the Born rule. In Section \ref{section_additivity_relations}, we have proved that additivity cannot be derived from the other two additional assumptions, non-contextuality and normalization. In Section \ref{section_additivity_in_5_proofs}, we have shown its crucial role in five important existing derivations of the Born rule. In Gleason's Theorem and Busch's extended Gleason Theorem, it is the core postulate. In Deutsch's derivation, it is implicitly assumed and also plays a central role throughout the whole proof. In Zurek's derivation, an attempt is made to weaken the additivity assumption, but this leads to the problem of continuity. In Hartle's derivation, the lack of an additivity assumption causes incoherence in his main idea when mixed states are considered. To sum up, we not only prove the necessity of the additivity assumption in these five derivations, but also show that, up to now, there is no equivalent assumption that can substitute for it.

Additivity is an assumption with a clear probabilistic nature, while non-contextuality and normalization are not. Our results facilitate the goal of determining the minimal set of additional postulates required to derive the Born rule, and also help to better understand the origin of probability in quantum mechanics. In the previous literature, some authors argue that additional assumptions beyond the postulates of standard quantum mechanics are essential for deriving the Born rule \cite{vaidman2020_Derivations_Born_Rule}, and some even believe that these additional assumptions must have probabilistic properties \cite{Barnum_2000_Quantum_probability_from_decision_theory}. Our work provides more solid support for these positions.

We have also given an organized set of quantum mechanics postulates in Section \ref{section_preliminaries} --- it clearly shows which assumptions are regarded as additional. It is also worth noting that different derivations can use different postulates from non-probabilistic quantum mechanics. For example, the State postulate (QM1) is not used in Gleason's or Busch's proofs, but is used in the other three; the Observable postulate (QM3) or the POVM postulate (QM3') is used in all other four derivations except Zurek's, and as a compensation, Zurek's derivation is also the only one that assumes the Evolution postulate (QM2). This difference is intriguing and deserves further exploration.

In Section \ref{section_preliminaries} and Appendix \ref{Appendix_Additive_Function_and_Continuity}, we also give clear and distinguishable definitions for various types of additivity, non-contextuality, normalization, and continuity. In the literature we discuss, ambiguity in these concepts can easily cause confusion, leading people to incorrectly conclude that two assumptions are equivalent, or that one method can be applied to prove a property in another situation. Our clarification helps to eliminate these confusions. This fills a gap in the literature and is beneficial for related research in the field.

Our analyses of the five derivations are mostly based on the first version of each method. This is because we include many different methods of deriving the Born rule, and it is impossible to examine every version of every method in a single paper. Although these initial versions may contain loopholes that are better addressed in later work, they represent the origins of the method chains and are the versions with the clearest physical intuitions. Discussion of these versions is the most important first step in further investigating each method. In the future, we plan to discuss Wallace's proof of the Deutsch–Wallace Theorem \cite{wallace2009}, and the improved version of Finkelstein-Hartle Theorem by Farhi, Goldstone, and Gutmann \cite{Farhi1989}. We also plan to discuss derivations of the Born rule proposed by supporters of Bohmian mechanics \cite{D_rr_1992} and QBism \cite{DeBrota_2021}. 

Readers may have noticed the importance of non-contextuality in the derivations we discussed. We do not provide a more complete discussion of the role of non-contextuality in this work, but this point will also be addressed in subsequent works.

\section{Acknowledgement}

We thank Vlatko Vedral for valuable discussions and support on this work. We are also grateful to Chiara Marletto for inspiring discussions regarding the Deutsch-Wallace Theorem, and to Fabrizio Logiurato for helpful clarifications regarding their work on non-contextuality in the Born rule.

\onecolumngrid



\bibliography{ref.bib}

\newpage
\onecolumngrid

\section*{Appendix}

\makeatletter
\renewcommand{\p@subsection}{}
\makeatother

\subsection{Equivalence between Gleason's Additivity Assumption and the Strong Normalization}

We prove that, for projective operators, the additivity assumption (or the frame function assumption) used by Gleason \cite{gleason1957_original_GT} is equivalent to strong normalization when $W=1$.

The definition of a frame function $\mu$ of weight $W$ can be written in the following equivalent form: if a countable set $\{\ket{i}\}_i$ forms an orthonormal basis of $\mathcal{H}$ --- the definition of a separable Hilbert space guarantees the existence of such bases --- then

$$
\sum_i \mu(\ket{i}\bra{i})=W
$$

where $W$ is a fixed real value for a given $\mu$ and $\mc{H}$.

To prove our result, we first introduce two lemmas.

\begin{lemma}

We can decompose any orthogonal projection $A$ on a separable real or complex Hilbert space into a sum of disjoint rank-one orthogonal projections $\ket{i}\bra{i}$, i.e.,

$$A = \sum_{i \in \mc{I}_A} \ket{i}\bra{i}$$
    
\end{lemma}

\begin{proof}

An orthogonal projection $A$ always has eigenvalues equal to 0 or 1. Therefore, it can be transformed into a diagonal matrix, where the diagonal entries with value 1 correspond to the $\ket{i}\bra{i}$'s.

\end{proof}

\begin{lemma}

\label{lemma_GT_sub_frame_f}

\cite{gleason1957_original_GT} If we have a frame function $\mu$ on $\mathcal{H}$ with weight $W$, then it is also a (sub-)frame function when restricted to any closed subspace of $\mathcal{H}$. The weight of the sub-frame function is usually different from $W$.

\end{lemma}

\begin{proof}

We can fix one $\ket{i_k}$. Then

$$
\sum_{i \ne i_k}{f(\ket{i})} = W- f(\ket{i_k})
$$

The right-hand side of the equation is constant for all choices of $\{ \ket{i} \}_{i \ne i_k}$, which form bases of the closed subspace. By induction, we can prove this result for any subspace of any dimension.

\end{proof}

\begin{lemma}

A function $\mu$ on projective operators is a frame function with weight $W=1$ if and only if it is strongly normalized.

\end{lemma}

\begin{proof}

If $\mu$ is strongly normalized, then by definition it is a frame function. Conversely, if $\mu$ is a frame function, to test whether $\mu$ is strongly normalized, we need to define $\mu(A)$ for arbitrary projective operator $A$. Using the additivity property of the frame function and Lemma \ref{lemma_GT_sub_frame_f}, it is natural to define

$$\mu(A) = W_A = W - \sum_{j \notin \mc{I}_A} \mu(\ket{j}\bra{j}) = \sum_{i \in \mc{I}_A} \mu(\ket{i}\bra{i})$$

Since $W=1$, $\mu$ is strongly normalized.

\end{proof}

\begin{remark}

Although ANC is often assumed together with additivity, both Gleason's additivity and strong normalization admit versions that do not require ANC.

\end{remark}

\subsection{Zurek's Proof for Observable Postulate}
\label{Appendix_Zurek_proof_QM3}

Zurek does not assume the Observable postulate (QM3). Instead, he derives it from the State postulate (QM1), the Evolution postulate (QM2), and the Eigenstate postulate (QM4). We rephrase his proof as follows \cite{Zurek_2018}.

Suppose the only two possible eigenstates of a quantum state $\ket{\psi}$ measured by an observable $X$ are $\ket{x_1}$ and $\ket{x_2}$. From the Many-Worlds Interpretation (MWI) and (QM2), the measurement process is described by an entangling unitary evolution $U_X$. Then, by the Eigenstate postulate (QM4),

$$U_X\ket{x_1}\ket{E_0} = \ket{x_1}\ket{E_1}$$

$$U_X\ket{x_2}\ket{E_0} = \ket{x_2}\ket{E_2}$$

where $\ket{E_0}$ denotes the initial state of the measurement device, and $\ket{E_1}$ and $\ket{E_2}$ denote the states corresponding to $\ket{x_1}$ and $\ket{x_2}$, respectively. Then,

$$\bra{E_0}\bra{x_1}U_X^\dagger U_X\ket{x_2}\ket{E_0} = \bra{E_1}\langle x_1 | x_2 \rangle\ket{E_2}$$

which implies

$$\langle x_1 | x_2 \rangle ( 1- \langle E_1 | E_2 \rangle)=0$$

Zurek then concludes that if we do not want $\langle E_1 | E_2 \rangle = 1$, which would mean that the outcomes are indistinguishable, we must have $\langle x_1 | x_2 \rangle = 0$, namely, the eigenstates are orthogonal. This argument applies to any two distinct eigenstates in any dimension. Thus, the orthogonality of eigenstates is established \cite{Zurek_2018}.

If we further assume that the eigenvalues are real values, then the Observable postulate (QM3) follows. This assumption is natural, since all readings on a measurement device are real values.

\subsection{Zurek's Weak Additivity Assumption and the Usual Countable Additivity}
\label{Appendix_Zurek's_Additivity}

Zurek shows how his weak additivity (E1) leads to our familiar countable additivity as follows \cite{zurek2005}.

In the situations of interest, events are described by measurement outcomes. For equal-amplitude states, if there are $n$ eigenstates, we know from Zurek's preceding proof that $\mu(x_i) = 1/n$ for any outcome $x_i$.

By Zurek's Weak Additivity (E1),

$$\mu(x_1+...+x_{n-1}) + \mu(x_n) = 1$$

$$\mu(x_1+...+x_{n-1}) = 1 - \frac{1}{n}$$

Zurek then claims that the conditional probability satisfies

$$\mu(x_1+...+x_{n-2} ~|~ x_1+...+x_{n-1}) = 1 - \frac{1}{n-1}$$

We clarify the reasoning here: since the envariant swaps still apply to the outcomes $x_1, ..., x_{n-1}$ when $x_n$ is ruled out, we can still conclude that $\mu(x_1) = \mu(x_2) = ... = \mu(x_{n-1}) = \frac{1}{n-1}$. Then,

$$\mu(x_1+...+x_{n-2}) = \left(1- \frac{1}{n}\right) \left(1- \frac{1}{n-1} \right) = \frac{n-2}{n}$$

We can permute $x_i$'s arbitrarily and repeat this multiplication multiple times. This yields

$$\mu( \sum_{j \in \mc{I}_m} x_j) = \frac{m}{n}$$

where $m=|\mc{I}_m|$.

Thus, we obtain the additivity property for states with equal amplitudes \cite{zurek2005}. We can easily extend this additivity to states with unequal but rational probabilities for the outcomes using the auxiliary system method. That is, if the outcome events $\{x_i\}_i$ are mutually disjoint and each $\mu(x_i) \in \mbb{Q}$,

$$\mu \left (\sum_i x_i \right) = \sum_i \mu \left (x_i \right) $$

However, there is no way to extend this additivity to states with real probabilities using this method. Therefore, Zurek's additivity is weaker than Gleason's additivity or the hidden countable additivity in Deutsch's proof. We cannot guarantee the continuity of $\mu$ from Zurek's Weak Additivity assumption, even when combined with non-negativity (a counterexample is presented in Section \ref{subsection_Zurek_proof}).

\subsection{Measure, Additive Function, and Continuity}
\label{Appendix_Additive_Function_and_Continuity}

We discuss the continuity issue in this section. This clarifies why questions about continuity frequently appear in derivations of the Born rule. First, we introduce the definition and some important properties of measures and additive functions from mathematics. Second, we distinguish between five different notions of continuity that appear in the literature (Remark \ref{remark_diff_continuity}): although Gleason's theorem starts from a measure, the notion of continuity used in its proof is not the same as the usual continuity discussed in measure theory \cite{gleason1957_original_GT,folland2013real}. These two are also very different from the continuity used in Busch's derivation \cite{busch2003}. Continuity for Busch's derivation, together with the continuity required in Deutsch's \cite{deutsch1999} and Zurek's proofs \cite{zurek2005}, are closer to each other and to the continuity for additive functions \cite{Kuczma2009_chap5}, but they are still not identical. This clarification is useful because it is easy to confuse these different notions of continuity and mistakenly assume that a proof for one kind of continuity will work for another. Third, we explain why an additive function can be discontinuous. The conditions for an additive function to be continuous are given in Theorem \ref{Thm_continuity_of_additive_function} and Theorem \ref{Thm_continuity_at_x_to_general_continuity} --- this helps explain why Gleason introduces the non-negativity assumption to ensure continuity \cite{gleason1957_original_GT}. Finally, in Theorem \ref{thm_continuity_by_busch} we show how Busch's additivity implies his continuity (a step that is not explicit in his original paper) \cite{busch2003}. Our proof makes clear how the flexibility of scaling POVM effects makes a simple proof for continuity possible. This flexibility is absent for additivity restricted to projective operators, as used in Gleason's, Deutsch's, and Zurek's derivations.

\begin{definition}

\cite{folland2013real} Let $X$ be a set and $\Sigma$ a $\sigma$-algebra on $X$. A measure $\mu$ on $X$ is a function $\mu:\Sigma \rightarrow \mathbb{R} \cup \{+\infty\}$ such that

(i) If $\{ A_i \}$ is a sequence of disjoint sets in $\Sigma$, then $\mu\big(\cup^\infty_{i=1} A_i\big) = \sum^\infty_{i=1} \mu(A_i)$. (countable additivity)

(ii) $\mu(\emptyset) = 0$

(iii) $\mu \geq 0$ 

\end{definition}

\begin{definition}

\cite{Kuczma2009_chap5} An additive function is a function $f: \mathbb{R}^N \rightarrow \mathbb{R}$ that satisfies 

$$f(\ket{x}+\ket{y})=f(\ket{x})+f(\ket{y})$$

for any $\ket{x},\ket{y}\in \mathbb{R}^N$. We will also write these vectors as $x,y$ for simplicity. This equation is called the Cauchy functional equation.

\end{definition}

\begin{remark}

We should note a subtle but important difference in the definitions of measure and additive function: the codomain of a measure may include $+\infty$, while the codomain of an additive function does not. Therefore, properties of additive functions do not always apply to measures. In Gleason's \cite{gleason1957_original_GT} and Busch's \cite{busch2003} derivations, which are based on measures or probability measures, the normalization and non-negativity assumptions are essential to rule out the output value $+\infty$.

\end{remark}

Also, for our purposes, we cannot always directly use results for additive functions, since our function $\mu$ takes an operator $A$ as input rather than a vector $\ket{x}$. In the end, we will find that $\mu$ is linear with respect to any $\ket{x}\bra{x}$, rather than with respect to $\ket{x}$. We analyze this in more detail in Remark \ref{remark_mu_and_additive_function} below.

We now list some important properties of additive functions.

\begin{lemma}

If $\mu$ is an additive function, then $\mu$ satisfies finite additivity, but not necessarily countable additivity.

\end{lemma}

\begin{proof}

It is immediate from the definition that an additive function satisfies finite additivity. For countable additivity, there is a well-known counterexample --- Let $A$ be a subset of the natural numbers $\mathbb{N}$. Define $\mu(A)=0$ if $|A|$ is finite and $\mu(A)=1$ if $|A|$ is infinite. This function $\mu$ is finitely additive but not countably additive \cite{ash2000probability}.

\end{proof}

\begin{lemma}
\label{Thm_additive_f_linear_for_q}

\cite{Kuczma2009_chap5} For an additive function $f$,

$$f(qx) = q  f(x)$$

for any $q \in \mathbb{Q}$ and $x \in \mathbb{R}^N$. 

\end{lemma}

\begin{proof}

From finite additivity, we have $f(my)=mf(y)$ for $m \in \mbb{N}$ and $y\in \mbb{R}^N$.

Let $y=x/n$ for $n \in \mbb{N}$. Then

$$f(m \frac{x}{n})=mf( \frac{x}{n}) = \frac{m}{n}f(x) ~~ (m,n \in \mbb{N},x,y\in \mbb{R}^N)$$

$$f(qx)=qf(x)  ~~ (q=\frac{m}{n} \in \mbb{Q},x\in \mbb{R}^N)$$

\end{proof}

Notice Lemma \ref{Thm_additive_f_linear_for_q} shows that an additive function is linear for any rational coefficient, but it is not always linear for real coefficients and thus is not always a linear function.

We will show how an additive function can be nonlinear and discontinuous in Lemma \ref{lemma_additive_f_not_continuous}. To do that, we first distinguish between five different continuities.

\begin{remark}
\label{remark_diff_continuity}

There exist different definitions of continuity in the literature which indicate very different properties. It is therefore easy to confuse them. We list five different notions of continuity closely related to our topic below:

(1) \textbf{Continuity for an additive function:} An additive function $\mu: \mbb{R}^N \rightarrow \mathbb{R}$ is continuous if 

$$\lim_{x\rightarrow y} f(x)-f(y) = 0$$

Here the equal sign indicates the $(\epsilon$-$\delta)$ definition of limit. We omit an explanation of this elementary definition.

(2) \textbf{Continuity for Busch's derivation \cite{busch2003} from rational to real coefficients:} A real-valued function $\mu(A)$ (where $A$ is a POVM effect) is continuous if for any $q\in \mathbb{Q}$ and $r \in \mathbb{R}$, 

$$\lim_{q\rightarrow r} \mu(q A)-\mu(r A) = 0 $$

Certainly, for genuine continuity we should consider arbitrary $q,r \in \mbb{R}$. However, for the purpose of extending the linearity of $\mu$ from rational to real coefficients in front of $A$, the notion of continuity discussed here is enough. A proof of continuity for Busch’s derivation from his postulates is given below in Theorem \ref{thm_continuity_by_busch}.

(3) \textbf{Continuity required in Deutsch's \cite{deutsch1999} and Zurek's \cite{zurek2005} derivations:} The value function $\mu$ is continuous if, for two different quantum states $\ket{\psi_1}$ and $\ket{\psi_2}$, in which a common eigenstate $\ket{x_i}$ has amplitudes $q \in \mbb{Q}$ and $r \in \mathbb{R}$ in the two states, respectively, we always have

$$\lim_{q\rightarrow r} \mu(\ket{\psi_1},\ket{x_i}\bra{x_i}) -\mu(\ket{\psi_2},\ket{x_i}\bra{x_i}) = 0 $$

These three continuities for additive functions, Busch's derivation, and Deutsch's and Zurek's proofs, are related but distinct. Note that in Busch's notion the coefficients $q,r$ multiply a POVM effect $A$, while in Deutsch's and Zurek's notion the coefficients $q,r$ are amplitudes in front of an eigenstate $\ket{x_i}$.

(4) \textbf{Gleason's continuity for a frame function:} Let $\mu$ be a frame function. $\mu$ is continuous if for any $x$ we can find a neighborhood $W$ of $x$ such that \cite{gleason1957_original_GT}

$$\sup\{f(y) ~|~ y\in W\} - \inf \{f(y) ~|~ y\in W\} \leq \epsilon \quad (\epsilon \rightarrow 0)$$

In Gleason's paper \cite{gleason1957_original_GT} he actually mentions two different continuity notions. The first one appears at the beginning of the paper when he discusses the existence of discontinuous additive functions. The second one is the notion we show here, which he uses as a crucial step of the proof: any continuous frame function on $\mbb{R}^3$ is regular (Theorem 2.3 in \cite{gleason1957_original_GT}).

(5) \textbf{Continuity of measure concerning inclusion of sets:} In measure theory, a measure $\mu$ is continuous from below if, for sets $A_1 \subset A_2 \subset \cdots$, we have \cite{folland2013real}

$$\mu(\cup_{i=1}^{\infty} A_i) = \lim_{j \rightarrow \infty} \mu(A_j)$$

The definition for continuity from above is analogous.

There is a classical result in measure theory that if $\mu$ is countably additive, it is also continuous in this sense \cite{folland2013real}. However, this continuity is very different from continuities for Gleason's or Busch's derivation, which are the types used to derive the Born rule. Gleason's continuity concerns how $\mu$ changes when we slightly rotate $\ket{x}$, while the continuity of measure concerns how $\mu$ changes when we enlarge the set $A=\ket{x}\bra{x}$.

\end{remark}

We can now show how an additive function can be discontinuous. The following theorems are important results from the theory of additive functions.

\begin{theorem}
\label{Thm_Hamel_basis_for_R}

\cite{Kuczma2009_chap4} $\mbb{R}^N$ is a vector space over the field $\mbb{Q}$. Its basis is called a Hamel basis, which has cardinality $\aleph_1$ and cannot be written down explicitly. There are many different Hamel bases, as long as the basis vectors are linearly independent.

\end{theorem}

\begin{theorem}
\label{Thm_additive_f_from_Hamel_basis}

\cite{Kuczma2009_chap5} If $H$ is a Hamel basis of $\mbb{R}^N$, then for any function $g:H \rightarrow \mbb{R}$ there exists a unique additive function $f: \mbb{R}^N \rightarrow \mbb{R}$ such that

$$f(h_i) = g(h_i) ~~~~ h_i \in H$$

$f$ is continuous if and only if for every $h_i \in H$, we have $\frac{g(h_i)}{h_i} = C_0$, where $C_0$ is a constant. 

\begin{proof}

The proof is omitted here and can be found in \cite{Kuczma2009_chap5}. 

\end{proof}

\end{theorem}

\begin{lemma}
\label{lemma_additive_f_not_continuous} 
    
An additive function is not always continuous. 

\end{lemma}

\begin{proof}

From Theorem \ref{Thm_Hamel_basis_for_R} and Theorem \ref{Thm_additive_f_from_Hamel_basis}, we can construct a Hamel basis for $\mbb{R}^N$ which contains two basis vectors $h_1,h_2$ such that $h_1 \rightarrow h_2$. They are guaranteed to be Hamel basis vectors as long as $\frac{h_1}{h_2} \notin \mbb{Q}$. Define

\begin{align*}
g(h_1) &= C_1~h_1 \\
g(h_2) &= C_2~h_2 \\
\end{align*}

where $C_1,C_2\in \mbb{R}$. We can let $C_1 << C_2$, for example, $C_1=1, C_2=10000$. Then we can construct the additive function $f$ based on $g$. For this $f$ one has

$$h_1 \rightarrow h_2, ~~~~ \text{but} ~~f(h_1) = h_1 << f(h_2) = 10000 h_2$$

So additive function $f$ is not continuous.

\end{proof}

\begin{remark}

The measure function $\mu$ used by Gleason takes operators as inputs, not vectors. This is why Gleason introduces another additive function $g:\mbb{R} \rightarrow \mbb{R}$ and claims that $g \circ \mu$ is still a frame function and can be wildly discontinuous \cite{gleason1957_original_GT}. All such pathological frame functions are unbounded (see Theorem \ref{Thm_continuity_of_additive_function} below). Gleason then introduces the non-negativity assumption to guarantee boundedness and thus continuity.

\end{remark}

We now present the conditions for an additive function to be a linear functional.

\begin{theorem}
\label{Thm_continuity_of_additive_function}

\cite{Kuczma2009_chap5} An additive function $f:\mbb{R}^N \rightarrow \mbb{R}$ has the form

$$f(\ket{x}) = \langle c | x\rangle = \sum_{i=1}^N c_i x_i$$

if it satisfies any of the conditions below:

(a) $f$ is continuous. 

(b) $f$ is continuous at some point $x$. 

(c) $f$ is bounded below or above on some interval.

(d) $f \ge 0$ or $f \le 0$ for all $x$.

\end{theorem}

\begin{proof}

(a) If $N=1$, since we already know $f(qx) = qf(x)$ when $q\in \mbb{Q}$ from Lemma \ref{Thm_additive_f_linear_for_q}, continuity implies $f(rx) = rf(x)$ for all $r \in \mbb{R}$. Let $c=f(1)$; then $f(x)=xf(1)=cx$ for any $x\in \mbb{R}$ \cite{Kuczma2009_chap5}.

If $N>1$, because of additivity we can write $f(\ket{x}) = f(x_1)+\cdots+f(x_N)$. From the one-dimensional result, each $f(x_i)=c_i x_i$, so $f(\ket{x}) = \langle c | x\rangle = \sum_{i=1}^N c_i x_i$ \cite{Kuczma2009_chap5}.

(b) Continuity at a single point extends to the whole domain (see Theorem \ref{Thm_continuity_at_x_to_general_continuity} below). The result follows from (a).

(c) The full proof is omitted here. The key idea is that any additive function is a convex function, which satisfies

$$f(\frac{x+y}{2}) \leq \frac{f(x)+f(y)}{2}$$

for all $x,y\in \mbb{R}^N$ \cite{Kuczma2009_chap5}. It can be shown that a convex function that is bounded above or below at any point is bounded above or below everywhere, and that a bounded convex function is continuous. Readers can see \cite{Kuczma2009_chap6} for more details. The result then follows from (a).

(d) The result follows directly from (c).

\end{proof}

\begin{theorem}
\label{Thm_continuity_at_x_to_general_continuity}

If an additive function $f$ is continuous at some $x \in \mbb{R}^N$, then it is continuous on the whole domain.

\end{theorem}

\begin{proof}

From Lemma \ref{Thm_additive_f_linear_for_q}, $f(0) = 0$. If $f$ is continuous at a point $x$, then by additivity

$$\lim_{\sigma\rightarrow \vec{0}}f(x+\sigma)=f(x)+\lim_{\sigma\rightarrow \vec{0}}f(\sigma).$$

Hence $\lim_{\sigma\rightarrow \vec{0}}f(\sigma) = 0$, so $f$ is continuous at $\vec{0}$. It is then trivial to prove $f$ is continuous at any point.

\end{proof}

\begin{remark}
\label{remark_mu_and_additive_function}

If $f(\ket{x}) = \langle c | x \rangle$, $f$ is a linear functional . However, we cannot directly apply this result to derive the Born rule. We would like to point out two important differences between the additive function $f$ and our real-valued function $\mu$ here: first, $f$ takes $\ket{x}$ as input, and not all operators $A$ in the domain of $\mu$ can be written as $A = \ket{x}\bra{x}$. Only projective operators have such property. Second, by definition, an additive function $f$ takes $\ket{x} \in \mbb{R}^N$ as input, whereas the eigenstates of $A$ may contain complex entries.

Thus, for our purpose, we can do the following modifications: any POVM effect $A$ can be written as $A = \sum_i \eta_i \ket{x_i}\bra{x_i} ~ (\eta_i \geq 0)$. Because of additivity assumed separately, we have

$$\mu(A) = \sum_i \eta_i \mu (\ket{x_i}\bra{x_i})$$

From Theorem \ref{Thm_continuity_of_additive_function}, it is easy to show that if an additive function $f$ takes $\ket{x}\bra{x}$ as input and satisfies conditions such as continuity, then

$$f(\ket{x}\bra{x}) = \langle x |   \left( \sum _j \lambda_j |c_j \rangle \langle c_j | \right) |x \rangle ~~~~~~~ (\lambda_j \ge 0)$$

This is also true when we extend the domain of $\ket{x}$ to $\mbb{C}^N$. Clearly, $\left( \sum _j \lambda_j |c_j \rangle \langle c_j | \right) \geq 0$. We can let $\rho = \left( \sum _j \lambda_j |c_j \rangle \langle c_j | \right)$ and rescale it to obtain a trace-one operator. Combining the two equations above yields the usual Born rule,

$$\mu(\rho,A) = \tr[\rho A]$$

Note that this is not a rigorous proof. We only intend to explain the idea of how one can obtain the function $\mu$ for the Born rule from an additive function $f$ when their domains differ. A rigorous proof is just Busch's derivation \cite{busch2003}.

\end{remark}

We now show how continuity for Busch's derivation can be derived from Busch's postulates: countable additivity, normalization, and non-negativity \cite{busch2003}. Note that this proof does not work for Gleason's derivation \cite{gleason1957_original_GT}, because Gleason's countable additivity degenerates to finite additivity when the Hilbert space is finite dimensional.

\begin{theorem}
\label{thm_continuity_by_busch}

If real-valued function $\mu(A)$ ($A$ is any POVM effect) satisfies the following conditions:

(i) If effects $A+B+...\le I$, $\mu(A+B+...) = \mu(A) +\mu(B) + ...$.

(ii) $\mu(I) =1 $.

(iii) For any $A \ge 0$, $\mu(A) \ge 0$.

Then $\mu$ satisfies the continuity for Busch's derivation, i.e., for any $q\in \mathbb{Q}$ and $r \in \mathbb{R}$, 

$$\lim_{q\rightarrow r} \mu(q A)-\mu(r A) = 0 $$

\end{theorem}

\begin{proof}

We first prove $\lim_{\sigma\rightarrow0}\mu(\sigma  A)=0$ for any POVM effect $A$. From countable additivity, we have

$$\mu(A)=\mu(\sum_{i=1}^{\infty}\frac{A}{2^i})=\sum_{i=1}^{\infty} \mu(\frac{A}{2^i})$$

Since $\mu(A) \in \mbb{R}$, the series on the right-hand side must converge. This guarantees that

$$\lim_{i\rightarrow +\infty} \mu(\frac{A}{2^i})=0$$

and hence

$$\lim_{\sigma \rightarrow 0} \mu(\sigma A) = 0$$

Notice this is not true for $\mu$ that is only finitely additive. In that case $\lim_{\sigma \rightarrow 0} \mu(\sigma A)$ can be nonzero because there is no infinite terms in the sum.

We can now prove the continuity of $\mu$. To prove that $\mu$ is continuous at $r A$ ($r \in \mbb{R}$), we need to show:

(i) $p(r A)$ exists for any $r \in \mbb{R}$. 

(ii) $p(r A)=\lim_{q \rightarrow r } p(qA)$, where $q \in \mbb{Q}$. 

The first step (i) is trivial since $r A$ is an effect. For (ii), if $q$ approaches $r$ from below, we write $rA=qA+\sigma A$ with $\sigma\ge 0$. By additivity,

$$p(r A)=p(qA+\sigma  A)=p(qA)+p(\sigma  A)$$

Thus

$$\lim_{q\rightarrow r} p(qA)=p(r A)- \lim_{\sigma \rightarrow 0} p(\sigma  A) = p(r A)$$ 

A similar proof applies if $q$ approaches $r$ from above. Therefore $\mu$ is continuous in the stated notion.

\end{proof}

\end{document}